\documentclass[copyright,creativecommons]{eptcs}
\providecommand{\event}{ACT 2020} 
\usepackage{underscore}           

\usepackage{amsthm}
\usepackage{amsmath}
\usepackage{tikz}
\usepackage{tikz-cd}
\usepackage{amsfonts}
\usepackage{amssymb}
\usepackage{enumitem}
\usepackage{mathtools}
\usepackage{wrapfig}
\usepackage{stmaryrd}
\usepackage{bbm}
\usepackage{underscore}
\usepackage{graphicx}
\usepackage{subcaption}
\usepackage{commath}
\usetikzlibrary{backgrounds}
\input{insbox}
\usepackage{mathtools}

\newtheorem{theorem}{Theorem}

\newtheorem{lemma}[theorem]{Lemma}

\theoremstyle{definition}
\newtheorem{definition}[theorem]{Definition}
\theoremstyle{remark}
\newtheorem{example}[theorem]{Example}

\newtheorem*{examples}{Examples}

\newcommand*\diff{\mathop{}\!\mathrm{d}}
\newcommand*{\bind}{\mathrel{\scalebox{0.7}[1]{$>\!\!\!>\!=$}}}
\newcommand*{\defeq}{\stackrel{\text{def}}=}

\title{A Monad for Probabilistic Point Processes}
\author{
Swaraj Dash \qquad\qquad Sam Staton
\institute{University of Oxford\\
Oxford, United Kingdom}
\email{\quad swaraj.dash@cs.ox.ac.uk \quad\qquad sam.staton@cs.ox.ac.uk}
}
\def\titlerunning{A Monad for Point Processes}
\def\authorrunning{S. Dash \& S. Staton}
\begin{document}
\maketitle

\begin{abstract} 
A point process on a space is a random bag of elements of that space. In this paper we explore
programming with point processes in a monadic style. To this end we identify point processes
on a space $X$ with probability measures of bags of elements in $X$. We describe this view of
point processes using the composition of the Giry and bag monads on the category of measurable
spaces and functions and prove that this composition also forms a monad using a distributive
law for monads. Finally, we define a morphism from a point process to its intensity measure, 
and show that this is a monad morphism. A special case of this monad morphism gives us Wald's
Lemma, an identity used to calculate the expected value of the sum of a random number of
random variables. Using our monad we define a range of point processes and point process
operations and compositionally compute their corresponding intensity measures using 
the monad morphism.
\end{abstract}

\section{Introduction}\label{intro}
Point processes (e.g.~\cite{ppbook}) are random collections of points. They serve as important tools in statistical analysis, where they are used to study various phenomena in fields as diverse as ecology, astronomy, computational neuroscience, and telecommunications, and in Bayesian analysis, where they are used for probabilistic inference (e.g.~\cite{nips1}).
\begin{figure}[b]
\centering
\includegraphics[scale=0.13]{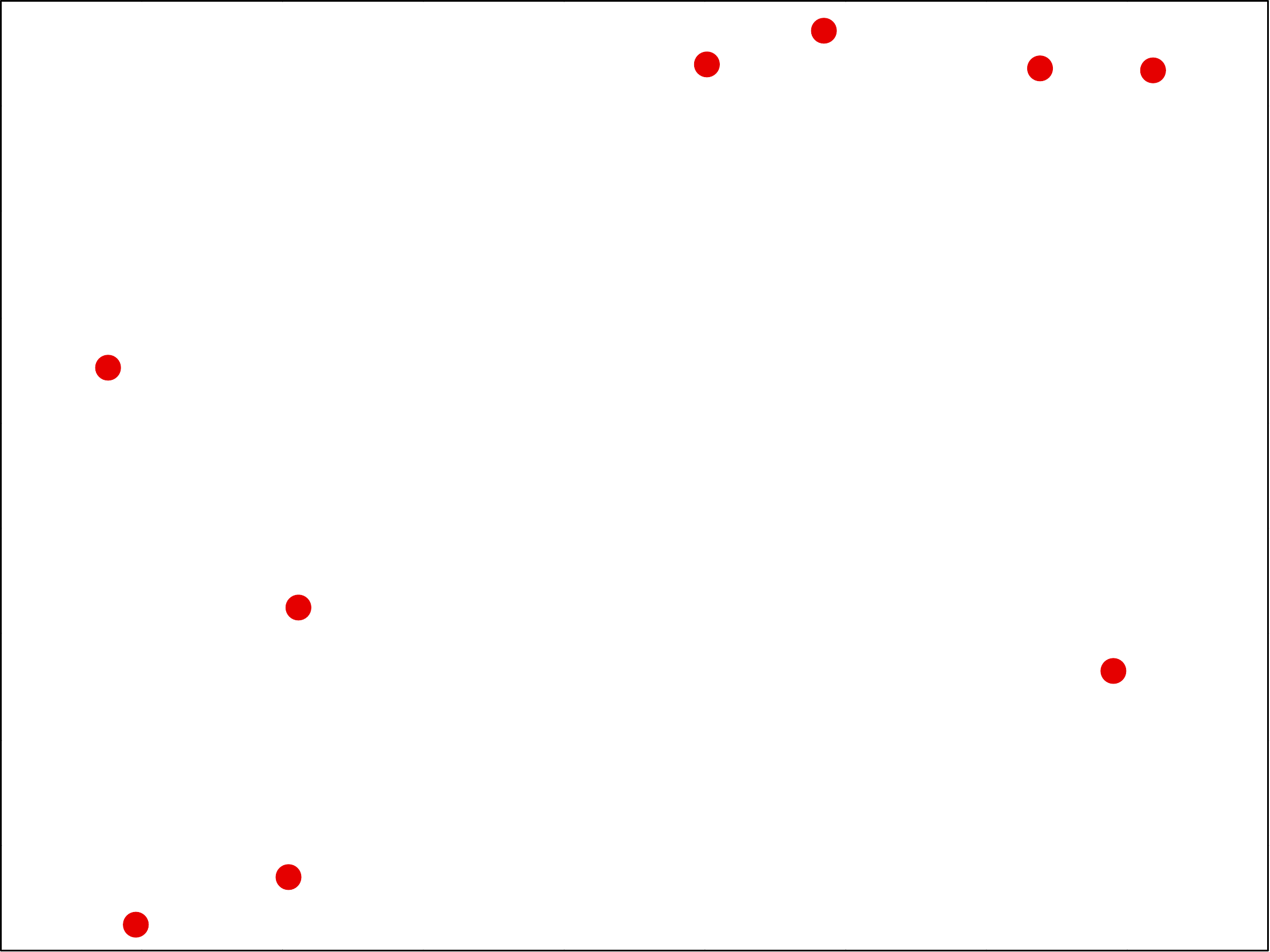}
\includegraphics[scale=0.13]{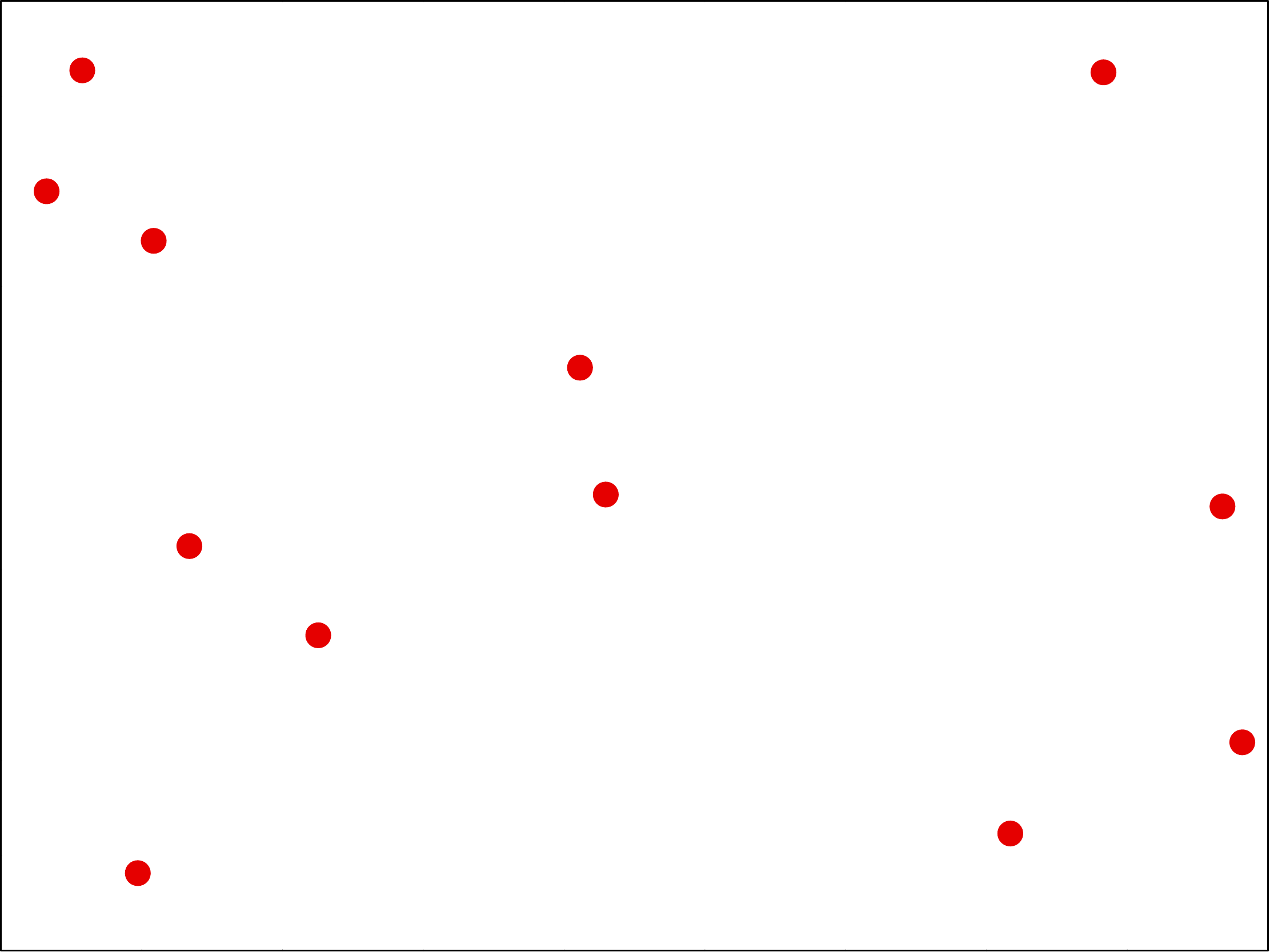}
\includegraphics[scale=0.13]{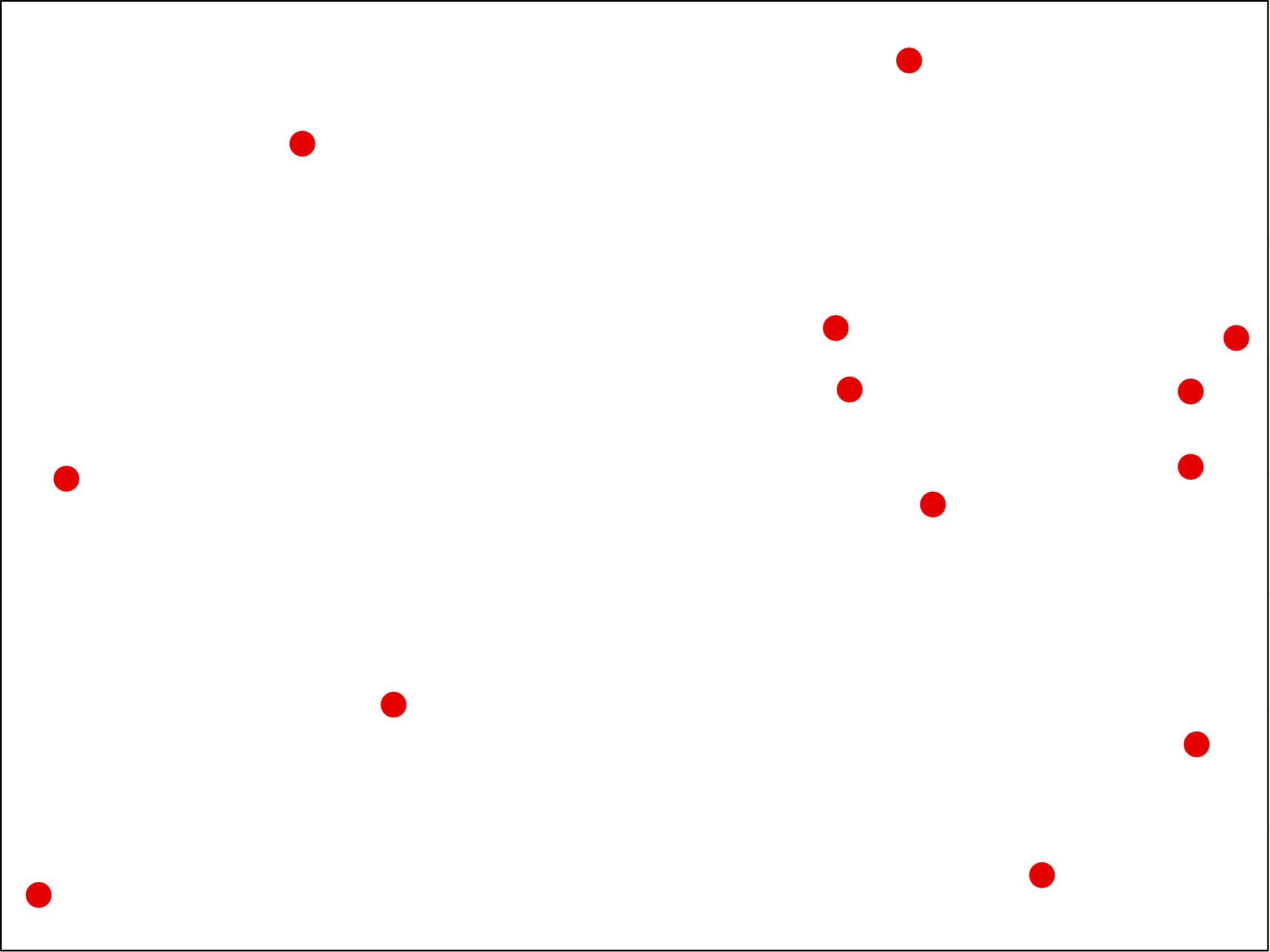}
\includegraphics[scale=0.13]{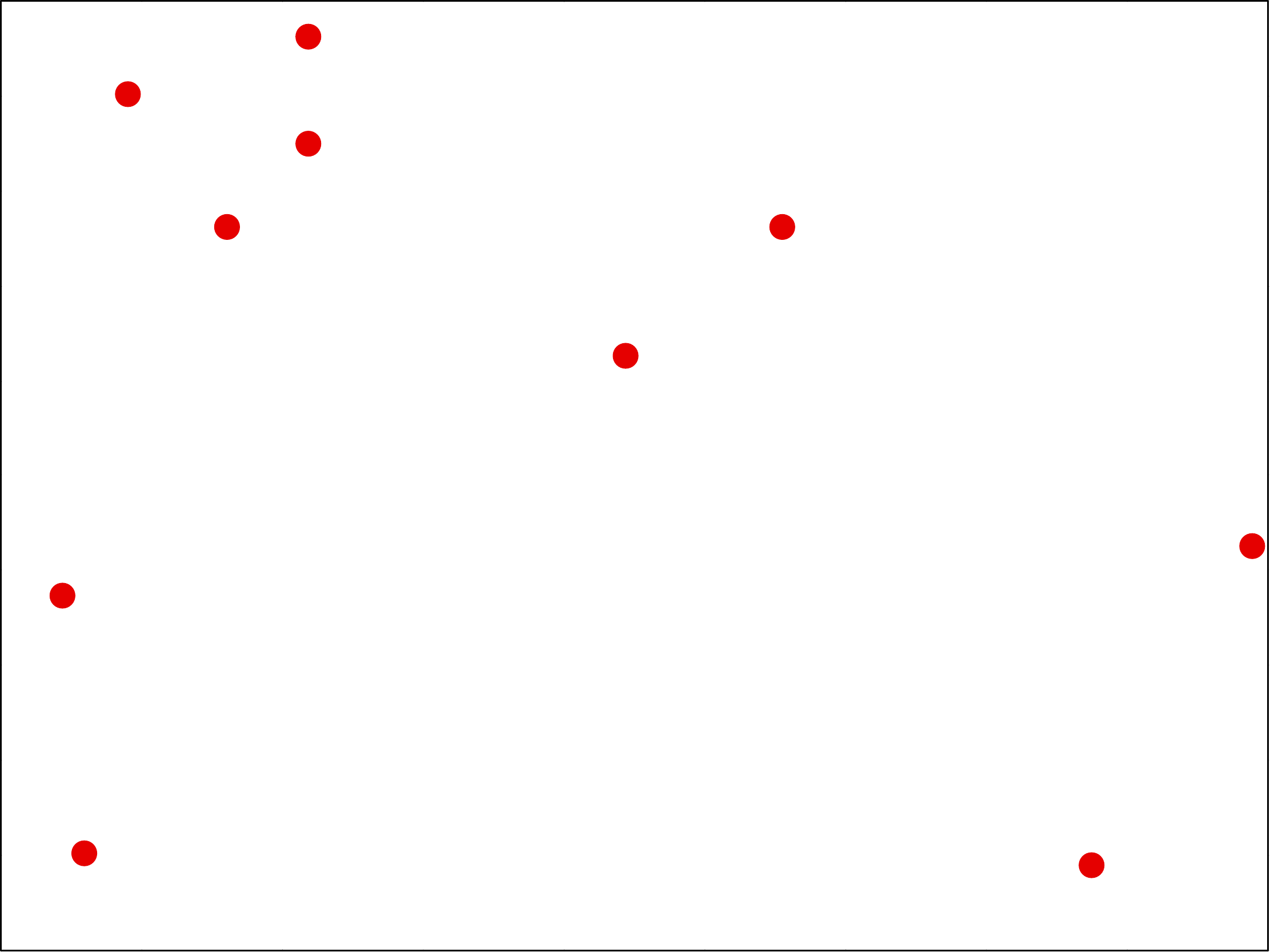}
\includegraphics[scale=0.13]{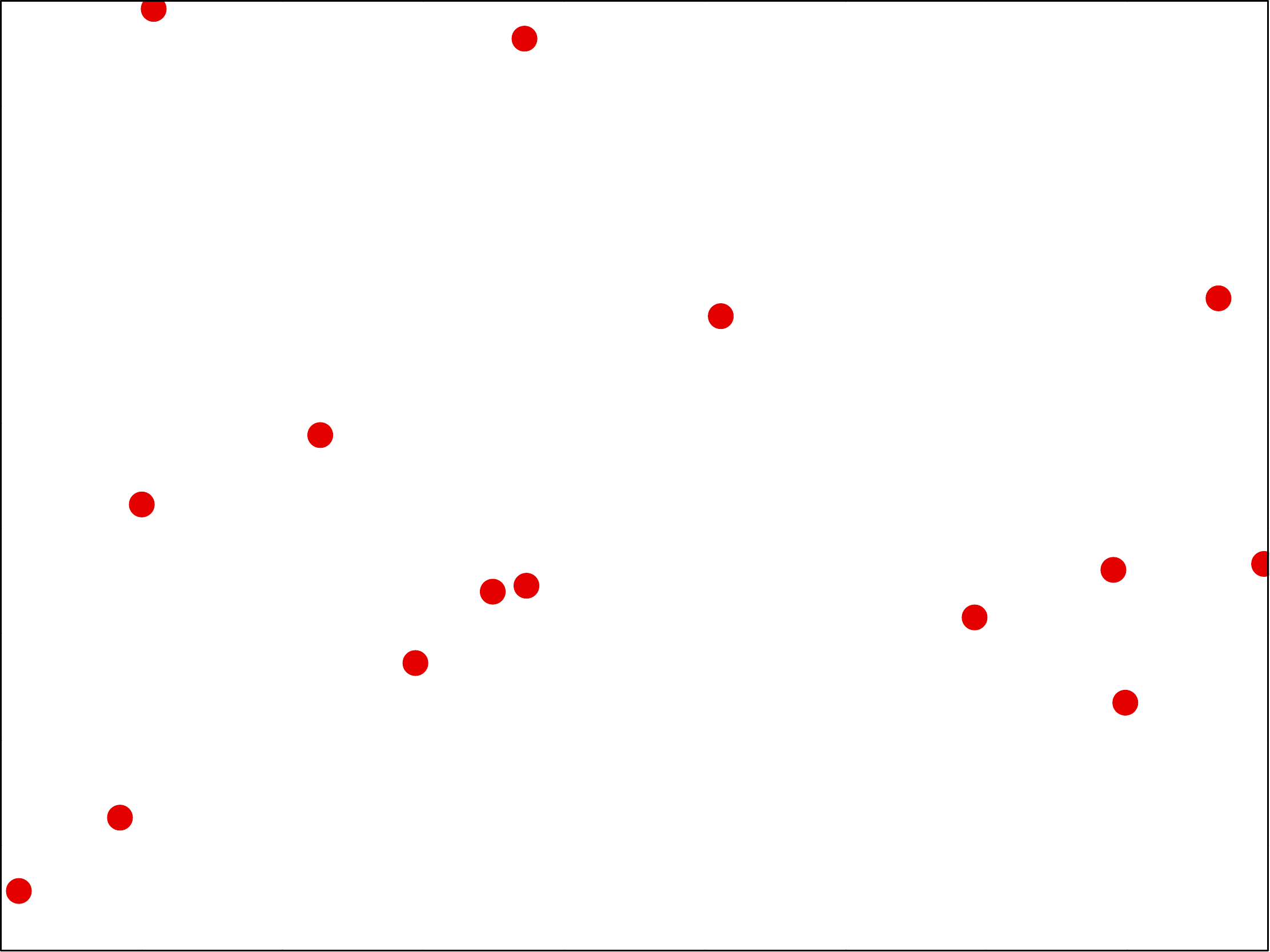}
\caption{Five draws from a Poisson point process on the unit square with rate 10.0.
\label{example}}
\end{figure}%
 As a simple example, in Figure~\ref{example} we illustrate five draws from a Poisson point process on the unit square. A Poisson point process is defined to be one in which the number of points in any two disjoint regions are independent of
each other.
One of the core tools of point process theory is the notion of intensity measure, which assigns to each region the average number of points that will appear in the region. In a homogeneous Poisson point process like Figure~\ref{example}, the average number of points in a region is proportional to the area of the region. This is a very simple point process, but an important starting point for many models.

The centerpiece of our categorical analysis of point processes is the space
\[
  G(B(X))
\]
which we now explain.
\begin{itemize}
  \item $X$ is a measurable space such as the unit square~$\mathbb I^2$ or the natural numbers~$\mathbb N$. We work in a category of measurable spaces so that we can discuss probability and integration in both the countable and uncountable settings. (See \S\ref{sec:prelims} for details.)
  \item $B(X)$ is a space of all bags of points in $X$. A bag (aka multiset) is a finite unordered list of elements in $X$. For example, the first draw in Figure~\ref{example} is a bag of 9 points in $X=\mathbb I^2$,
    and the last draw in Figure~\ref{nat-example} is a bag of 6 points in $X=\mathbb N$, with 5 overlapping points at $0$ (multiplicity $5$) and one at $6$ (multiplicity~$1$).
    (See \S\ref{sec:bag} for details.)
  \item $G(B(X))$ is a space of point processes, i.e.~the space of probability measures on the space of bags of points in $X$. Here, $G$ stands for Giry, who carried out early work on the category theory of spaces of probability measures.
  \end{itemize}
  (This space $G(B(X))$ is typically uncountable,  but as is common in statistics, it is helpful to run a simulation, outputting a finite number of draws, as in Figures~\ref{example}, \ref{nat-example}, and~\ref{bind}.) 

  This paper has two main contributions:
  \begin{itemize}
  \item     The construction $G(B(X))$ forms a monad (\S\ref{pp-monad}). This is useful because it gives us a compositional framework for point processes. One can build point processes (elements $1\to G(B(X))$) by composing morphisms in the Kleisli category for the monad $GB$, using a syntax like Haskell's do-notation (\S\ref{sec:examples}). 
Our construction of a monad uses Beck's theory of distributive laws of monads. 
\item
  The construction assigning to each point process its intensity measure is a monad morphism (\S\ref{sec:intensity}, Theorem.~\ref{thm:expmonadmorph}). Thus, if we build a point process in a compositional way, we can also calculate its intensity in the same compositional way. The key idea here is to regard both $G$ and $B$ as submonads of a monad $M$ of all measures, so that the intensity measure function can be defined as a composite
  \[
    GB\hookrightarrow MM\xrightarrow \mu M\text.
  \]
  \end{itemize}
\vspace{-4mm}


%
\begin{figure}[t]
\centering
\begin{tikzpicture}[framed,scale=0.5,every node/.style={scale=0.9}]
\draw (0,0.1) -- (0,-0.1);
\draw (1,0.1) -- (1,-0.1);
\draw (2,0.1) -- (2,-0.1);
\draw (3,0.1) -- (3,-0.1);
\draw (4,0.1) -- (4,-0.1);
\draw (5,0.1) -- (5,-0.1);
\draw (6,0.1) -- (6,-0.1);
\draw (-0.5,0) -- (6.5,0);
\node at (1,-0.6) {$1$};
\fill[white] (4,3.5) circle (3pt);
\node at (2,-0.6) {$2$};
\node at (3,-0.6) {$3$};
\node at (4,-0.6) {$4$};
\node at (5,-0.6) {$5$};
\node at (0,-0.6) {$0$};
\node at (6,-0.6) {$6$};
\fill[red] (4,0.5) circle (3pt);
\fill[red] (1,0.5) circle (3pt);
\fill[red] (0,0.5) circle (3pt);
\fill[red] (3,0.5) circle (3pt);
\fill[red] (5,0.5) circle (3pt);
\end{tikzpicture}
\begin{tikzpicture}[framed,scale=0.5,every node/.style={scale=0.9}]
\draw (0,0.1) -- (0,-0.1);
\draw (1,0.1) -- (1,-0.1);
\draw (2,0.1) -- (2,-0.1);
\draw (3,0.1) -- (3,-0.1);
\draw (4,0.1) -- (4,-0.1);
\draw (5,0.1) -- (5,-0.1);
\draw (6,0.1) -- (6,-0.1);
\draw (-0.5,0) -- (6.5,0);
\node at (1,-0.6) {$1$};
\node at (2,-0.6) {$2$};
\node at (3,-0.6) {$3$};
\node at (4,-0.6) {$4$};
\fill[white] (4,3.5) circle (3pt);
\node at (5,-0.6) {$5$};
\node at (0,-0.6) {$0$};
\node at (6,-0.6) {$6$};
\fill[red] (1,0.5) circle (3pt);
\fill[red] (0,0.5) circle (3pt);
\fill[red] (2,0.5) circle (3pt);
\end{tikzpicture}
\begin{tikzpicture}[framed,scale=0.5,every node/.style={scale=0.9}]
\draw (0,0.1) -- (0,-0.1);
\draw (1,0.1) -- (1,-0.1);
\draw (2,0.1) -- (2,-0.1);
\draw (3,0.1) -- (3,-0.1);
\draw (4,0.1) -- (4,-0.1);
\draw (5,0.1) -- (5,-0.1);
\draw (6,0.1) -- (6,-0.1);
\draw (-0.5,0) -- (6.5,0);
\node at (1,-0.6) {$1$};
\node at (2,-0.6) {$2$};
\node at (3,-0.6) {$3$};
\node at (4,-0.6) {$4$};
\node at (5,-0.6) {$5$};
\node at (0,-0.6) {$0$};
\node at (6,-0.6) {$6$};
\fill[white] (4,3.5) circle (3pt);
\fill[red] (2,1.0) circle (3pt);
\fill[red] (2,1.5) circle (3pt);
\fill[red] (2,2.0) circle (3pt);
\fill[red] (2,0.5) circle (3pt);
\fill[red] (4,1.0) circle (3pt);
\fill[red] (4,1.5) circle (3pt);
\fill[red] (4,0.5) circle (3pt);
\fill[red] (1,0.5) circle (3pt);
\fill[red] (3,0.5) circle (3pt);
\end{tikzpicture}
\begin{tikzpicture}[framed,scale=0.5,every node/.style={scale=0.9}]
\draw (0,0.1) -- (0,-0.1);
\draw (1,0.1) -- (1,-0.1);
\draw (2,0.1) -- (2,-0.1);
\draw (3,0.1) -- (3,-0.1);
\draw (4,0.1) -- (4,-0.1);
\draw (5,0.1) -- (5,-0.1);
\draw (6,0.1) -- (6,-0.1);
\draw (-0.5,0) -- (6.5,0);
\node at (1,-0.6) {$1$};
\node at (2,-0.6) {$2$};
\node at (3,-0.6) {$3$};
\node at (4,-0.6) {$4$};
\fill[white] (4,3.5) circle (3pt);
\node at (5,-0.6) {$5$};
\node at (0,-0.6) {$0$};
\node at (6,-0.6) {$6$};
\fill[red] (0,0.5) circle (3pt);
\fill[red] (6,0.5) circle (3pt);
\fill[red] (0,1.0) circle (3pt);
\fill[red] (0,1.5) circle (3pt);
\fill[red] (0,2.0) circle (3pt);
\fill[red] (0,2.5) circle (3pt);
\end{tikzpicture}
\caption{Four draws from a point process on the natural numbers.\label{nat-example}}
\end{figure}
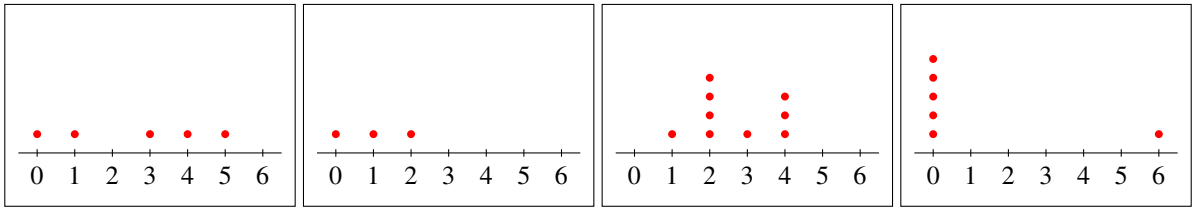

\paragraph{Broader context.}
The broader context of this work is the idea that category theory can be a language for organizing the structure of statistical models. At one end of the spectrum, this line of work involves a foundational categorical analysis (e.g.~\cite{dirichlet-is-natural,cones,fritz, support-morphism,garner,quasiborel,jacobs-book,mccullagh,simpson}). At the other end of the spectrum is `probabilistic programming’, a popular method of statistical modelling using programs (e.g.~\cite{church,smc,hakaru}); in many instances this is functional programming and so heavily inspired by category theory. This full spectrum of work plays a foundational role in probability and statistics, but also addresses a practical grand challenge of interpretability in statistical models, since category theory and programming allow us to clearly organize the structure of complicated statistical models, via composition. 

Our work here appears to be the first work on point processes in this context. However, point processes are widely used in statistical models in practice. We mention two programming styles whose relationship to point processes has only recently become evident:
\begin{itemize}
\item  probabilistic logic programming, in the style of \textsc{Blog}~\cite{russell}, is about describing random sets of points;
\item probabilistic databases are random bags of records~\cite{grohe-lindner}.
\end{itemize}
In future work we intend to set out in more detail how our categorical development in this paper can be used to inform the practice of probabilistic logic programming and probabilistic databases in a compositional way.


\section{Mathematical preliminaries}
\label{sec:prelims}
We recall basic measure theory, which is the standard formulation for probability theory over uncountable spaces (e.g.~\cite{measure-prob}). 
\subsection{Measure theory}
\begin{definition}
A \emph{$\sigma$-algebra} on a set $X$ is a nonempty family $\Sigma_X$ of subsets
of $X$ that is closed under complements and countable unions. 
The pair $(X,\Sigma_X)$ is called a \emph{measurable space} (we just write $X$ when 
$\Sigma_X$ can be inferred from context).

Given $(X,\Sigma_X)$, a \emph{measure} is a function
$\nu: \Sigma_X \rightarrow \mathbb{R}^\infty_+$ such that
for all countable collections of disjoint sets $A_i \in \Sigma_X$,
$\nu\left(\bigcup_i A_i\right) = \sum_i \nu(A_i).$
In particular,  
$\nu(\varnothing) = 0$. It is a \emph{probability measure} if $\nu(X)=1$.
A \emph{pre-measure} is defined to be  this additivity condition
except that it is not necessarily defined on a $\sigma$-algebra. 
\end{definition}

\begin{examples}
The Borel sets form the least $\sigma$-algebra $\Sigma_\mathbb{R}$ of subsets of $\mathbb{R}$
that contain the intervals $(a,b)$.

\noindent
On a countable set $X$, such as $\mathbb N$ or the one-point set $\mathbbm 1 = \{ \star \}$, we will typically consider the \emph{discrete} $\sigma$-algebra, which contains all the subsets. 
In this context, the measures are entirely determined by their values on singletons, $\nu(\{x\})$, and so a measure is the same thing as a function $X\to \mathbb R^\infty_+$. 
\end{examples}
\begin{definition}
Let $(X,\Sigma_X)$ and $(Y,\Sigma_Y)$ be two measurable spaces. A \emph{measurable function}
$f : X \rightarrow Y$ is a function such that $f^{-1}(U) \in \Sigma_X$ when $U \in \Sigma_Y$.
The category {\bf Meas} contains as objects measurable spaces with the morphisms being
measurable functions between them.

For any measurable function $f:X\to \mathbb R^\infty_+$, and any measure $\nu:\Sigma_X\to \mathbb R^\infty_+$, we can define the \emph{Lebesgue integral} or \emph{expected value} $\int_X f\,\diff\nu$ of~$f$. 
\end{definition}
\begin{definition}
A measurable space $(X,\Sigma_X)$ is a \emph{standard Borel space} if it is either measurably
isomorphic to $(\mathbb{R},\Sigma_\mathbb{R})$ or it is countable and discrete.
\end{definition}
\noindent

\subsection{Giry monad}
The Giry functor $G: \text{\bf Meas} \rightarrow \text{\bf Meas}$ sends a measurable space to
the space of all possible probability measures on it~\cite{giry}. By slight abuse of
notation, let $G(X,\Sigma_X) := (GX,\Sigma_{GX})$.
$GX$ is the set of all probability measures $\nu : \Sigma_X \rightarrow [0,1]$ on $X$ 
equipped with the $\sigma$-algebra
$\Sigma_{GX}$ generated by the set of all evaluation maps
$\text{ev}_U : GX \rightarrow [0,1]$, sending $\nu$ to $\nu(U)$
(where $U \in \Sigma_X$). In other words, it is generated by sets of probability
measures $D^U_{I} = \{\nu : \Sigma_X \rightarrow [0,1] \mid \nu(U) \in I \}$.
$$\Sigma_{GX} = \sigma(\{ D^U_{I} \mid U \in \Sigma_X, I \in \mathcal{B}([0,1])\})$$
where $\mathcal{B}([0,1])$ is the Borel $\sigma$-algebra and $\sigma$ is the closure
operator which when given a family of subsets generates the required $\sigma$-algebra
by closing the family under countable unions and complements.
The following unit $\eta^G_X : X \to GX$ and multiplication $\mu^G_X : GGX \to GX$ make $G$ into a monad.
\begin{align*}
\eta^G_X &(x) = \delta_x 
= \lambda U\ldotp
  \begin{cases} 
   1 & \text{if } x \in U \\
   0 & \text{otherwise\footnotemark}
 \end{cases}
       &
\mu^G_X &(\nu)= \lambda U\ldotp \int_{GX} \text{ev}_U\ \diff\nu
\end{align*}
\footnotetext{Throughout this paper we use $\lambda$-notation to describe functions between sets.
  Note that the category $\mathbf{Meas}$ is not cartesian closed and so this is not intended as a formal internal language (c.f.~\cite{quasiborel}).}
Given measurable $f : X \rightarrow Y$, the functorial action $Gf: GX \rightarrow GY$ sends 
$\nu \in GX$ to the \emph{push-forward} measure of $\nu$ along $f$,
$\nu \circ f^{-1} \in GY$.
An important property of push-forward measures is the change-of-variables formula
$\int_Y g\, \diff (Gf)(\nu) = \int_X (g \circ f)\, \diff \nu$ (where $g : Y \to [0,1]$).

\subsection{All-measures monad}
\label{sec:allmeasuresmonad}
The all-measures functor $M : \textbf{Meas} \to \textbf{Meas}$ is defined similarly to the
Giry functor. It sends measurable spaces $(X,\Sigma_X)$ to the space of all 
measures $(MX,\Sigma_{MX})$ where $MX$ is the set of measures
$\mu : \Sigma_X \to \mathbb{R}^\infty_+$ and
$\Sigma_{MX} = \sigma(\{m^U_r \mid U \in \Sigma_X, r \in \mathbb{R}\})$ with
$m^U_r = \{ \mu: \Sigma_X \to \mathbb{R}^\infty_+ \mid \mu(U) < r\}$.
The same unit and multiplication maps make $M$ into a monad.
(Warning: The Giry monad is strong and commutative, but the all-measures monad $M$ is not strong, because Fubini's Theorem does not hold for arbitrary measures. So something more refined is needed for functional programming, but that is not an issue for this paper; see e.g.~\cite{sfinite} for details.)


\section{A monad for finite bags in {\bf Meas}}
\label{sec:bag}
In this Section we discuss the construction of the bag monad in the context of measure
theory. 
A bag (aka multiset) is a finite unordered list of elements in some set. For example,
the bag $[8,8,5,8,5]$ contains 5 twice and 8 three times, and can also be written as
$[5,5,8,8,8]$ since bags are unordered.
We begin by recalling the bag endofunctor $B$ in {\bf Set}, which we show to lift to 
an endofunctor
in {\bf Meas} by assigning the $\sigma$-algebra $\Sigma_{BX}$ to the space of bags in
$\S$~\ref{bag-functor} (Definition~\ref{bag-def}).
In $\S$~\ref{bag-monad} we prove that $B$, which is a monad in {\bf Set}, lifts to a monad
in {\bf Meas} by showing that the unit and multiplication maps extend to measurable functions.
Later in $\S$~\ref{pp-monad} we will need to define probability measures on $BX$, i.e.,
functions $\Sigma_{BX} \to [0,1]$. Defining such functions entails having to define
them on arbitrary combinations of unions and intersections of our generating sets $A^U_k$.
We simplify this task by making use of Carath\'{e}odory's extension theorem in
$\S$~\ref{sigma-bx} to show that it
suffices to simply define these functions on the generating sets of $\Sigma_{BX}$ without
needing to define them on all of $\Sigma_{BX}$.

\subsection{The bag functor in {\bf Set} and {\bf Meas}}\label{bag-functor}
Consider the well-known finite bag endofunctor $B : \textbf{Set} \to \textbf{Set}$ where
$BX$ is the set of all finite bags with elements of $X$. Given a function $f : X \to Y$,
the function $Bf : BX \to BY$ applies $f$ component-wise to its argument bag. The
natural transformations $\eta^B_X : X \to BX$ and $\mu^B_X : BBX \to BX$ which 
return the singleton bag and the (multiplicity respecting) union of bags
respectively make $(B,\eta^B,\mu^B)$ into a monad.
\begin{align*}
\eta^B_X &(x) = [x]          & \mu^b_X &([b_1,\dots,b_n]) = \bigcup_i b_i
\end{align*}
In order to lift $B : \textbf{Set} \to \textbf{Set}$ to $B : \textbf{Meas} \to \textbf{Meas}$
we equip a $\sigma$-algebra $\Sigma_{BX}$ to our set $BX$.

\begin{definition}[Measurable space of bags]\label{bag-def}
Let $BX$ be the set of bags on the measurable space $X$. Equip $BX$ with the $\sigma$-algebra $\Sigma_{BX}$
formed by the $\sigma$-closure of generating sets
$A^U_k = \{ b \in BX \mid \text{$b$ contains exactly}$ $\text{$k$ elements in $U$} \}$.
$$\Sigma_{BX} = \sigma(\{ A^U_k \mid U \in \Sigma_X, k \in \mathbb N\})$$
Then $(BX,\Sigma_{BX})$ is the measurable space of bags of $X$.
\end{definition}

It is important to note that the generating set $A^U_k \in \Sigma_{BX}$ contains bags
of $X$ of cardinality \emph{at least} $k$, as each bag in it contains $k$ elements in $U$,
in addition to possible other elements \emph{not} in $U$. The set $A^X_k$, on the other hand,
contains all the bags of $X$ of cardinality \emph{exactly} $k$ (since $X$ is our universal
set). Their intersection $A^X_n \cap A^U_k$ is then the set of
\emph{bags of cardinality $n$ with $k$ elements
in $U$}. This can be extended to construct
the set of \emph{bags of cardinality $n$ containing $k_i$ elements in $U_i$} for some family
of sets $U_i \in \Sigma_X$, which is then the intersection
$A^X_n \cap \left(\bigcap_i A^{U_i}_{k_i} \right)$.

\subsection{The bag monad in {\bf Meas}}\label{bag-monad}
\begin{lemma}\label{bag-mon-lem}
The unit and multiplication maps $\eta^B_X : X \rightarrow BX$ and
$\mu^B_X : BBX \rightarrow BX$ are measurable.
\end{lemma}
\begin{proof}
To prove the measurability of these functions it suffices to show that the pre-images
of the generating sets $A^U_k$ are measurable.
Consider $U \in \Sigma_{X}$ and some $A^U_k \in \Sigma_{BX}$. 
The inverse image map of the unit ${\eta^B_X}^{-1}(A^U_k)$ evaluates to $\overline U$
(the complement of $U$)
if $k=0$, $U$ if $k=1$, and $\varnothing$ otherwise, all of which are elements of 
$\Sigma_X$, and so ${\eta^B_X}^{-1}$ is measurable.
We now sketch why ${\mu^B_X}^{-1}(A^U_k) \in \Sigma_{BBX}$ and later show a detailed
argument for the case $k=4$. Call this set $\mathcal P$.
By definition of the inverse image map, $\mathcal P$ is the collection of bags of bags
such that the arbitrary union of each bag of bags contains exactly $k$ elements in $U$. 
Recall that the set of \emph{bags of cardinality $n$ containing $k_i$ elements in $U_i$}
for some family of sets $U_i \in \Sigma_X$, is given by
$A^X_n \cap \left(\bigcap_i A^{U_i}_{k_i} \right)$.
By using this technique of describing collections of bags
 and considering the various partitions of the number $k$ such that the resulting arbitrary
union of bags will contain $k$ elements in $U$, we can express $\mathcal P$ entirely
using measurable sets, allowing us to conclude that $\mu^B_X$ is a measurable function.
\end{proof}

\begin{example}
Let $\mathcal P = {\mu^B_X}^{-1}(A^U_4)$. 4 can be
partitioned in five ways:
$\{4, 3+1, 2+2, 2+1+1, 1+1+1+1\}.$
$\mathcal P$ is the set of bags of bags such that the arbitrary union of each
bag of bags contains exactly 4 elements in $U$. We start by considering elements
of this set based on their cardinalities.
\begin{itemize}
\item There is only one collection of bags of cardinality 1 which are members of $\mathcal P$.
These are the bags which contain a single bag which in turn contains 4 elements in $U$.
Denote this collection as $\langle 4 \rangle$.
\item There are three collections of bags of cardinality 2 which are members of $\mathcal P$.
The first contains two bags which have 4 and 0 elements in $U$, 
the second with 3 and 1 elements in $U$, and the third with
2 and 2 elements in $U$, respectively.  We write them as $\langle 4,0 \rangle$, 
$\langle 3,1 \rangle$, and $\langle 2,2 \rangle$.
\item Cardinality 3: $\langle 4,0,0 \rangle$, 
$\langle 3,1,0 \rangle$, $\langle 2,2,0 \rangle$, and $\langle 2,1,1 \rangle$.
\item Cardinality 4: $\langle 4,0,0,0 \rangle$, 
$\langle 3,1,0,0 \rangle$, $\langle 2,2,0,0 \rangle$, $\langle 2,1,1,0 \rangle$, 
$\langle 1,1,1,1 \rangle$.
\item Cardinality 5: $\langle 4,0,0,0,0 \rangle$, 
$\langle 3,1,0,0,0 \rangle$, $\langle 2,2,0,0,0 \rangle$, $\langle 2,1,1,0,0 \rangle$,  
$\langle 1,1,1,1,0 \rangle$. And so on.
\end{itemize}
Each collection is definable using the generating sets, and the collections of
different cardinalities are mutually disjoint. For example,
$\langle 3,1,0 \rangle = A^{BX}_3 \cap (A^{\mathbb{B}_3}_1 \cap A^{\mathbb{B}_1}_1 \cap
A^{\mathbb{B}_0}_1)$ and 
$\langle 2,1,1,0 \rangle = A^{BX}_4 \cap (A^{\mathbb{B}_2}_1 \cap A^{\mathbb{B}_1}_2 \cap
A^{\mathbb{B}_0}_1)$ 
 where $\mathbb{B}_i = A^{U}_i$. 
Finally, $\mathcal P$ is the union of all these disjoint collections.
$$\mathcal P = {\mu^B_X}^{-1}(A^U_4) = \langle 4 \rangle \cup \langle 4,0 \rangle \cup \langle 3,1 \rangle
\cup \langle 2,2 \rangle \cup  \langle 4,0,0 \rangle \cup \dots$$
\end{example}

\begin{theorem}
$(B : \textrm{\bf Meas} \to \textrm{\bf Meas},\eta^B,\mu^B)$ is a monad.
\end{theorem}
\begin{proof}
The monad laws hold as in {\bf Set}. Furthermore, $\eta^B$ and $\mu^B$ are measurable 
(Lemma~\ref{bag-mon-lem}).
\end{proof}

\subsection{Defining measures on ${BX}$}\label{sigma-bx}
In this Section we construct a \emph{ring} -- a set of sets containing the empty set
closed under pairwise unions and relative complements -- of the generating sets $A^U_k$ of
$\Sigma_{BX}$ in order to invoke Carath\'{e}odory's extension theorem. This allows us
to define measures by defining them on just specific unions of intersections of
the sets $A^U_k$ of $\Sigma_{BX}$ rather than having to define them on all the arbitrary
combinations of unions and intersections of these sets.

\begin{theorem}[Carath\'{e}odory's extension]
Let $\mathcal R$ be a ring and $\nu : \mathcal R \to \mathbb{R}^\infty_+$ be a pre-measure.
Then there exists a measure $\tilde\nu : \sigma(\mathcal R) \to \mathbb{R}^\infty_+$ such
that $\tilde\nu(S) = \nu(S)$ for all $S \in \mathcal R$.
\end{theorem}

We start by defining $\mathcal{R'}$ to be the set of countable intersections of our generating sets above such that their base sets are mutually disjoint.
$$\mathcal{R'} = \left \{ \bigcap_i A^{U_i}_{k_i} \middle \vert U_i \in \Sigma_X, \text {$U_i$'s mutually disjoint}, k_i \in \mathbb{N} \right \}$$
Now define $\mathcal{R}$ to be the closure of $\mathcal{R'}$ under countable unions.
The elements of $\mathcal{R}$ are the unions of intersections of certain generating sets. In particular, any $P \in \mathcal{R}$ can be expressed as
$P = \bigcup_i \bigcap_j A^{U_{i,j}}_{k_{i,j}}.$

\begin{examples}
One example of such a set $P$ is
$(A^\alpha_0 \cap A^\beta_2) \cup (A^\beta_2 \cap A^\gamma_1) \cup (A^\alpha_2 \cap A^\beta_5 \cap A^\delta_2)$
where $\alpha, \beta, \gamma, \delta \in \Sigma_X.$ Note that although $\alpha$ and $\beta$, $\beta$ and $\gamma$, and $\alpha$, $\beta$, and $\delta$ are all mutually disjoint (by definition of $\mathcal{R'}$), it is still possible for $\alpha$ and $\gamma$ to overlap. Using standard set theoretic identities we can redefine $P$ in terms of $\alpha \setminus \gamma \text{ (instead of just $\alpha$)}, \beta, \gamma, \delta$ for \emph{all} the base sets across the unions to be mutually disjoint. 

\noindent
Consider also the set
$(A^\alpha_1 \cap A^\beta_1) \cup (A^\alpha_0 \cap A^\gamma_2)$
where $\alpha, \beta, \gamma \in \Sigma_X$ are mutually disjoint. We can rewrite $A^\alpha_1 \cap A^\beta_1$ as $\bigcup_i (A^\alpha_1 \cap A^\beta_1 \cap A^\gamma_i)$ since $\bigcup_i A^\gamma_i$ is simply the universal set. The right half of the set above can similarly be rewritten, enabling us to reformulate it as
$\bigcup_i (A^\alpha_1 \cap A^\beta_1 \cap A^\gamma_i) \cup \bigcup_i (A^\alpha_0 \cap A^\beta_i \cap A^\gamma_2).$
\end{examples}

From the two examples above, we can assume without loss of generality that an arbitrary element $P \in \mathcal{R}$ will be of the form
$\bigcup_i \bigcap_j A^{U_j}_{k_{i,j}}$
such that all the $U_j$'s are mutually disjoint.
Finally, note that any two sets
$\bigcap_j A^{U_j}_{k_{m,j}}$ \text{and} $\bigcap_j A^{U_j}_{k_{n,j}}$
are disjoint unless for all $j$, $k_{m,j} = k_{n,j}$. And so, every $P$ can be viewed as the disjoint union of a set of sets. Call this the \emph{disjoint normal form} (it is not unique).

\begin{lemma} $\mathcal R$ contains the empty set and is closed under pairwise
unions and relative complements.
\end{lemma}
\begin{proof}
It is clear that $\varnothing \in \mathcal{R}$.
The set $\mathcal R$ is by definition closed under countable unions, and so is also closed under
pairwise unions.
Consider $P$ and $Q \in \mathcal{R}$ with their respective disjoint normal forms. 
Without loss of generality, we can express both $P$ and $Q$ using the same set of mutually
disjoint base sets $U_j \in \Sigma_X$. This gives us
$P = \bigcup_i \bigcap_j A^{U_j}_{a_{i,j}}$
and
$Q = \bigcup_i \bigcap_j A^{U_j}_{b_{i,j}}.$
Since both $P$ and $Q$ have been formed by taking the unions of a common set of disjoint sets belonging to $\mathcal{R'}$, their difference $P \setminus Q$ can also be expressed at the disjoint union of sets belonging to $\mathcal{R'}$ and so $P \setminus Q \in \mathcal{R}$.
\end{proof}

Having shown $\mathcal R$ to be a ring, we have by Carath\'{e}odory's extension theorem that
any pre-measure defined on $\mathcal R$ extends to a measure defined on $\sigma(\mathcal R)$.
And so, in order to define a measure on $BX$, it will suffice to define it on sets
$\bigcup_i \bigcap_j A^{U_j}_{k_{i,j}}$. We use this fact in the next Section.

\section{Point process monad}\label{pp-monad}
A point process on a space $X$ is a probability measure on bags of $X$. 
By composing the Giry and bag monads we can define $GBX$ to be the space of point
processes on $X$.  In other words, a point process $\alpha \in GBX$ is a
probability measure $\alpha : \Sigma_{BX} \to [0,1]$ assigning probabilities
to measurable subsets of bags $A^U_k$.
\emph{The probability of observing $k$ points
in the region $U$ of the point process $\alpha$ is then $\alpha(A^U_k)$.}

Earlier we showed that $G$ and $B$ both form monads.
It is well-known that the composition of two monads does not automatically yield a new monad. 
In this Section we prove that the composite functor $GB$ admits a monadic structure by
defining the natural transformation $l : BG \rightarrow GB$,
called the \emph{distributive law} \cite{distr} of $G$ over $B$, such that the
following identities hold.
\begin{align*}
\textbf{(Triangle I)}\ \ \ &
l \circ B\eta^G = \eta^GB
&
G\eta^B = l \circ \eta^BG
\ \ \ &\textbf{(Triangle II)}\\
\textbf{(Pentagon I)}\ \ \ &
l \circ B\mu^G = \mu^GB \circ Gl \circ lG
&
G\mu^B \circ lB \circ Bl = l \circ \mu^BG
\ \ \ &\textbf{(Pentagon II)} 
\end{align*}
This distributive law $l$ then induces the $GB$ monad with the unit defined as the 
horizontal composition $\eta^G \ast \eta^B$, and the join
defined as the composition of the horizontal composition $\mu^G \ast \mu^B$ with $GlB$.
$$ \eta^{GB} : 1 \xrightarrow{\eta^G \ast \eta^B} GB \qquad \text{and} \qquad
\mu^{GB} : GBGB \xrightarrow{GlB} GGBB \xrightarrow{\mu^G \ast \mu^B} GB $$

\subsection{Distributive law}
The distributive law $l_X : BGX \to GBX$ is a function from bags of probability measures
to probability measures on bags.
We showed in $\S$~\ref{sigma-bx} that in order to define a measure on $\Sigma_{BX}$
it suffices to define a pre-measure on sets of the form $\bigcup_i \bigcap_j A^{U_j}_{k_i,j}$. 
Carath\'{e}odory's extension theorem ensures this pre-measure extends to a measure.

In the definition that follows, we consider a bag of probability measures
$[\nu_1,\dots,\nu_n] \in BGX$ and $\bigcup_i \bigcap_j A^{U_j}_{k_i,j} \in \Sigma_{BX}$. 
We define the application of $l[\nu_1,\dots,\nu_n]$ to this disjoint union of intersections
as the sum of products of $l[\nu_1,\dots,\nu_n](A^{U_j}_{k_{i,j}})$. Each of these sub-terms
is in turn defined as
the push-forward of the product measure along $K_n$, a function mapping $n$-tuples to
bags of cardinality $n$. An example follows in (\ref{intuition}).
\begin{align*}
l[\nu_1,\dots,\nu_n]\left(\bigcup_i \bigcap_j A^{U_j}_{k_{i,j}}\right)\quad
&\defeq\quad \sum_i \prod_j l[\nu_1,\dots,\nu_n]\left(A^{U_j}_{k_{i,j}}\right)\\
&\defeq\quad \sum_i \prod_j GK_n(\otimes_i \nu_i)\left(A^{U_j}_{k_{i,j}}\right)
\end{align*}
where $K_n : Y^n \to B_nY$ is the measurable function that sends $n$-tuples
to bags of cardinality $n$ ($B_nY \subseteq BY$).
($K_n$ is measurable since
$K^{-1}_n : \Sigma_{B_nY} \to \Sigma_{Y^n}$ sends sets $A^U_k$
to their corresponding disjoint unions of $n$-products of $U$ and $\overline U$.
For example, $K^{-1}_2(A^U_1) = U \times \overline U\ \uplus\ \overline U \times U$.)
In the definition above $Y$ has been instantiated to be $GX$.

\emph{Intuition:} the term $l[\nu_1,\dots,\nu_n]$ is a point process where the probability
of observing $k$ points in some region $U \in \Sigma_X$ is the probability of observing a
total of $k$ points landing in $U$ after independently sampling a point each from all the $\nu_i$'s $\in GX$.
The following example calculation of $l[\nu_1,\nu_2](A^U_1)$ confirms this idea.
\begin{align}
	\label{intuition}
	\begin{split}
l[\nu_1,\nu_2](A^U_1) &= GK_2(\nu_1 \otimes \nu_2)(A^U_1) 
              = (\nu_1 \otimes \nu_2)(K^{-1}_2(A^U_1))  \\
              &= (\nu_1 \otimes \nu_2)(U \times \overline U\ \uplus\ \overline U \times U) 
              = \nu_1(U)\nu_2(\overline U) + \nu_1(\overline U)\nu_2(U).
	      \end{split}
\end{align}
Note that $l[\nu_1,\dots,\nu_n](A^U_k) = 0$ for $k > n$. Observe that
$\sum_i l[\nu_1,\dots,\nu_n](A^U_i) = 1.$

Showing that $l$ is measurable is a routine calculation, and is made simpler with the
knowledge that sets of constant-cardinality bags are measurable.

We noted earlier that the sets $A^U_m$ contain bags of varying cardinalities. In the following
lemma we show that the measure $l[\nu_1,\dots,\nu_n]$ acts \emph{only} on the subset of
these sets with cardinality $n$. This result is simple yet very useful in providing
an intuitive understanding of
the distributive law, and is instrumental in proving the second pentagon identity.
\begin{lemma}\label{card-lemma} For $[\nu_1,\dots,\nu_n] \in BGX$, 
$l[\nu_1,\dots,\nu_n](A^X_m \cap W) = l[\nu_1,\dots,\nu_n](W)$ if $m=n$ and 0 if $m \not= n$.
\end{lemma}
\begin{proof} Consider $l[\nu_1,\dots,\nu_n](A^X_m)$. Using the definition of $l$ this
probability can be expressed as the sum of products of a combination of $\nu_i(X)$ and
$\nu_j(\overline X)$ terms (where $i,j$ range from 1 to $n$). Unless $m=n$, each summand
will contain at least one factor
with $\nu_j(\overline X) = \nu_j(\varnothing) = 0$, nullifying the entire sum. 
It is only non-zero when $m=n$, representing the probability of observing $n$ points
in the entire space, which has probability 1. The measure of any set $A^X_n \cap W$
is then just the measure of $W$. 
\end{proof}

\begin{theorem} $(GB : \textbf{Meas} \to \textbf{Meas},\eta^{GB},\mu^{GB})$ is a monad
via the distributive law $l : BG \to GB$.
\end{theorem}
\noindent\textit{Proof (sketch).}
We prove that the four identities for the distributive law hold.
The two triangle identities follow from simple algebra.
For the first pentagon identity
we make use of the change-of-variables formula,
as $l[\nu_1,\dots,\nu_n] = GK_n(\bigotimes_i \nu_i)$ is a
push-forward measure, and prove the resulting equality using standard integration identities.
For the final pentagon identity we are required to work with the set ${\mu^B_X}^{-1}(A^U_k)$,
which we decompose using the method presented in Lemma~\ref{bag-mon-lem}. 
Invoking Lemma~\ref{card-lemma} on these constant-cardinality decompositions
allows us to simplify the resulting expression by removing sets with measure zero 
and prove the final equality after some more algebraic manipulations.
\qed

\subsection{Unit and bind}\label{gb-bind}
The unit $\eta^{GB}$ 
returns the deterministic point process $\eta^{GB}(x)$ with the singular point $x$.
When programming with monads it is often convenient to focus on Kleisli composition
in a stylized form, using the function
$\bind_{GB} : \mathbf{Meas}(X , GBY) \to \mathbf{Meas}(GBX,GBY)$ (pronounced \emph{bind});
we write $\alpha \bind_{GB} f$ for $\bind_{GB}(f)(\alpha)$~\cite{moggi}.

%

\begin{wrapfigure}[6]{r}{8cm}\vspace{-7mm}\begin{tikzpicture}[scale=1.8]
\draw (0,0) -- (0,1);
\draw (0,1) -- (1,1);
\draw (1,1) -- (1,0);
\draw (1,0) -- (0,0);
\fill[orange] (0.2,0.2) circle (1pt);
\fill[blue] (0.8,0.8) circle (1pt);
\node at (0.2,0.35) {\small $x_1$};
\node at (0.65,0.8) {\small $x_2$};
\node at (-0.25,0.5) {$\stackrel{(1)}\mapsto$};
\node at (1.25,0.5) {$\stackrel{(2)}\mapsto$};
\begin{scope}[shift={(1.5,0)}]
\draw (0,0) -- (0,1);
\draw (0,1) -- (1,1);
\draw (1,1) -- (1,0);
\draw (1,0) -- (0,0);
\fill[orange] (0.1,0.3) circle (1pt);
\fill[orange] (0.2,0.1) circle (1pt);
\fill[orange] (0.3,0.2) circle (1pt);
\fill[blue] (0.7,0.6) circle (1pt);
\fill[blue] (0.2,0.8) circle (1pt);
\fill[blue] (0.6,0.9) circle (1pt);
\fill[blue] (0.8,0.8) circle (1pt);
\fill[blue] (0.7,0.3) circle (1pt);
\node at (0.5,0.1) {\scriptsize $f(x_1)$};
\node at (0.5,0.75) {\scriptsize $f(x_2)$};
\node at (1.25,0.5) {$\stackrel{(3)}\mapsto$};
\end{scope}
\begin{scope}[shift={(3.0,0)}]
\draw (0,0) -- (0,1);
\draw (0,1) -- (1,1);
\draw (1,1) -- (1,0);
\draw (1,0) -- (0,0);
\fill[red] (0.1,0.3) circle (1pt);
\fill[red] (0.2,0.1) circle (1pt);
\fill[red] (0.3,0.2) circle (1pt);
\fill[red] (0.7,0.6) circle (1pt);
\fill[red] (0.2,0.8) circle (1pt);
\fill[red] (0.6,0.9) circle (1pt);
\fill[red] (0.8,0.8) circle (1pt);
\fill[red] (0.7,0.3) circle (1pt);
\end{scope}
\end{tikzpicture}
\caption{Sampling from a composite process.\label{fig:bind}}
\end{wrapfigure}
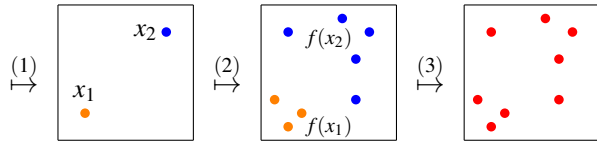%
This presents us
with useful intuition for programming with point processes.
Let $X$ and $Y$ be discrete sets.
Then the process of
sampling points from $\alpha \bind_{GB} f \in GBY$ can be viewed as the
following simulation, illustrated in Figure~\ref{fig:bind}.
(1)~Sample a bag of points from $\alpha$.
(2)~Each point $x_i$
in this bag produces a point process $f(x_i)$ from which we sample a bag of points in $Y$.
(3)~A sample of points from the overall point process is the union of these bags of points in $Y$.
This intuition allows us to declaratively program with point processes by
being able to define them simply by how they must be simulated.
We make extensive use of this intuition in the next Section.
\paragraph{Aside about related work.}
A long-term problem in programming semantics has been the combination of probability and non-determinism (e.g.~\cite{varacca}). In that context, it is well known that there is no distributive law between the probability monad and the powerset monad (e.g.~\cite{varacca,zwart,goy-petrisan}). It has recently become well-known that, in the set-theoretic case, it \emph{is} possible to find a distributive law by using a bag monad instead of a powerset monad (e.g.~\cite[App.~A]{keimel-plotkin}, \cite{parlant,zwart}). This is \emph{not} a sleight of hand, because point processes are not necessarily to be thought of in terms of computational non-determinism. In fact, in the probability theory literature, it is most common to regard point processes as random bags rather than as random sets.

The bag construction is a free commutative monoid, and the free commutative \emph{co}monoid also plays an important role in the theory of linear logic. Recently, bag-like exponentials have arisen in models of probabilistic linear logic \cite{dahlqvist-kozen,hamano}. The precise relationship to our monad and our distributive law remains to be seen.

More broadly, bags, multisets and urns play a fundamental role in statistics and arise at various points in a categorical treatment (e.g.~\cite{jacobs-book,jacobs-staton,jacobs-freq}).


\section{Examples of point processes via the monad}
\label{sec:examples}
\paragraph{Probability distributions as point processes}\label{cpd}
As a first example, we describe probability distributions on the natural numbers as point
processes on the singleton space $\mathbbm{1} = \{ \star \}$, based on the observation that a bag of singletons is a natural number 
($B\mathbbm{1} \cong \mathbb{N}$).
Any probability distribution $d \in G\mathbb{N}$ (so that $\sum_{i=0}^\infty d_i = 1$) can be
presented as a point process $\underline{d} \in GB\mathbbm{1}$ where we observe
$k$ copies of $\star$ with probability $d_k$:
\begin{equation}\underline{d}\left(A^{\{\star\}}_k\right) := d_k\label{eqn:GB1GBN}\end{equation}

\paragraph{Building compound probability distributions.}
Using our monad we define \emph{compound distributions} as point processes on
the unit type. 
A compound probability distribution is the probability distribution of the
sum of a number of independent identically-distributed random variables, where the number
of terms to be added is itself a random variable. For example, given a random variable
$N \sim \text{Poisson}(\Lambda)$ and iid variables $X_i$, the random variable
$Y = \sum_{i=0}^N X_i$
forms a compound Poisson distribution. 

Recall the behaviour of $\bind_{GB}$ on countable sets described in $\S$~\ref{gb-bind}.
By considering $\underline{N} \in GB\mathbbm{1}$ (say, the Poisson distribution) and
$\underline{X} \in GB\mathbbm{1}$ (the distribution of the iid $X_i$), we can express compound distributions as:
\begin{equation}\label{eqn:compounddist}
\gamma \ =\  \underline{N} \bind_{GB} \lambda \star\!.\  \underline{X} \qquad\in GB\mathbbm 1.
\end{equation}

\paragraph{A Poisson point process on the unit square}\label{pois-example}

The Poisson point process on the unit square $\mathbb{I}^2$ (Fig.~\ref{example}) can be simulated by first sampling a random number of points
from the Poisson distribution, and then uniformly distributing these points
across $\mathbb{I}^2$.
Consider again a Poisson distribution $\underline{N} \in GB\mathbbm{1}$ as a point process.
Now consider the point process $\underline{U} \in GB\mathbb{I}^2$ which returns a single
point uniformly distributed in $\mathbb{I}^2$. The Poisson point process $\pi$ can be built using the monad:
\begin{equation}\label{eqn:poisson-sim}
\pi \ =\ \underline{N} \bind_{GB} \lambda \star\!.\ \underline{U} \qquad\in GB\mathbb{I}^2. 
\end{equation}

\paragraph{Thinning a point process.}
Thinning is an operation applied to the points of an underlying point process, where
the points are thinned (removed) according to some probabilistic rule. Given some point process
$\alpha \in GBX$ and some thinning rule $t : X \rightarrow GBX$ such that $t(x)$
probabilistically returns either $[x]$ or $\varnothing$,
we can use the monad to build the thinned point process $\alpha' \in GBX$ as
\begin{align*}
\alpha' =  \alpha  \bind_{GB} \lambda x\!.\ t(x) \qquad\in GBX. 
\end{align*}
\paragraph{Displacing a point process.}
Displacement is an operation applied to the points of an underlying point process, where
the points are independently randomly displaced (translated) according to some distribution.
We model this distribution as a single-point point process $\Delta \in GB\mathbb R$.
The location of this random point is the random displacement distance.
For $\alpha \in GB\mathbb R$ we simulate the displaced point process
$\alpha' \in GB\mathbb{R}$ by sampling
an $x$ from $\alpha$, a displacement distance $d$ from $\Delta$, and then returning
the displaced point.
\[
\alpha' \ =\ \alpha \bind_{GB} \lambda x\!.\ 
(\Delta \bind_{GB} \lambda d\!.\  \eta^{GB}(x+d))\qquad \in GB\mathbb{R}. 
\]

\begin{wrapfigure}{r}{10cm}
\centering
	\includegraphics[scale=0.2]{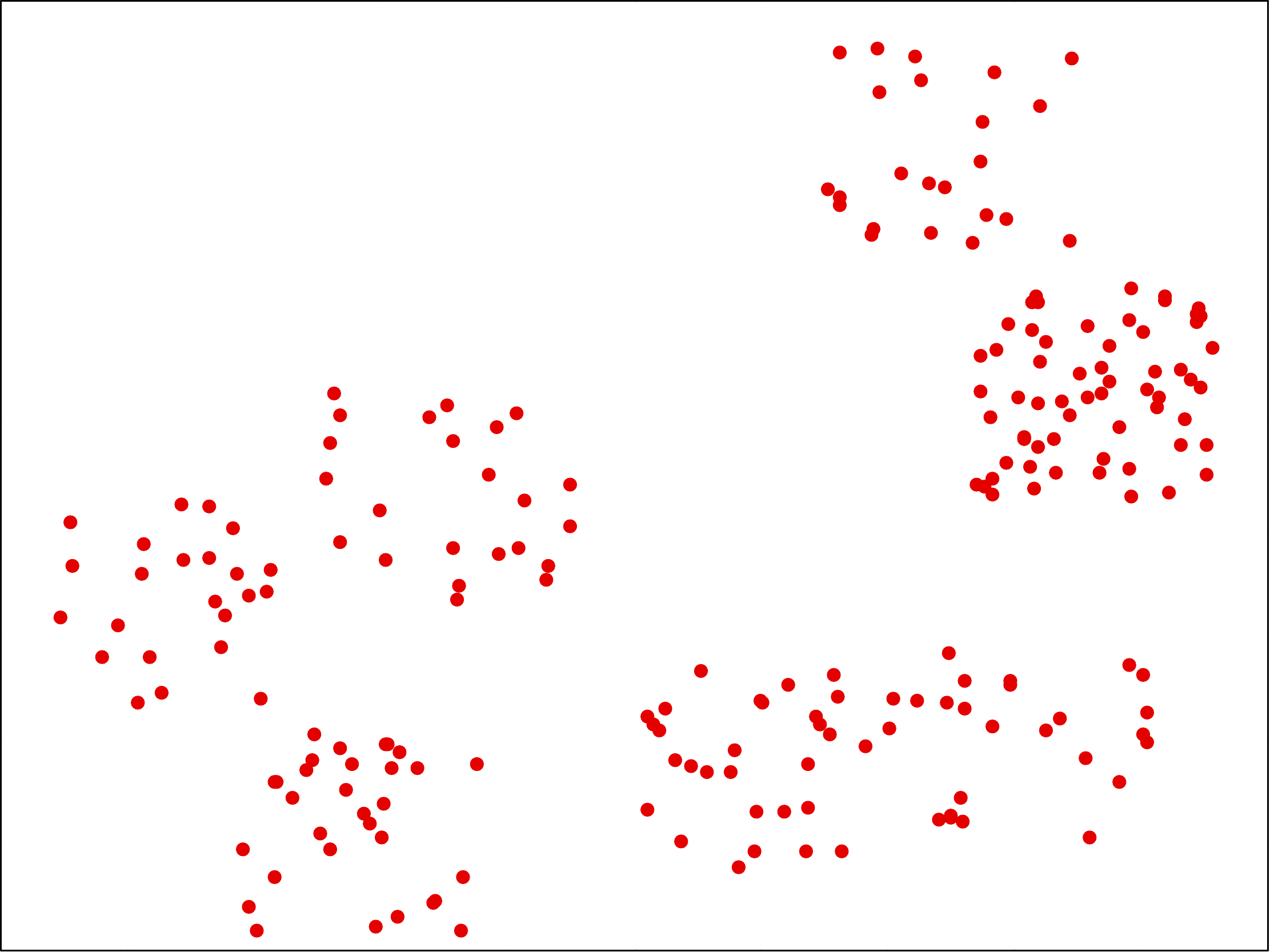}
\qquad
	\includegraphics[scale=0.2]{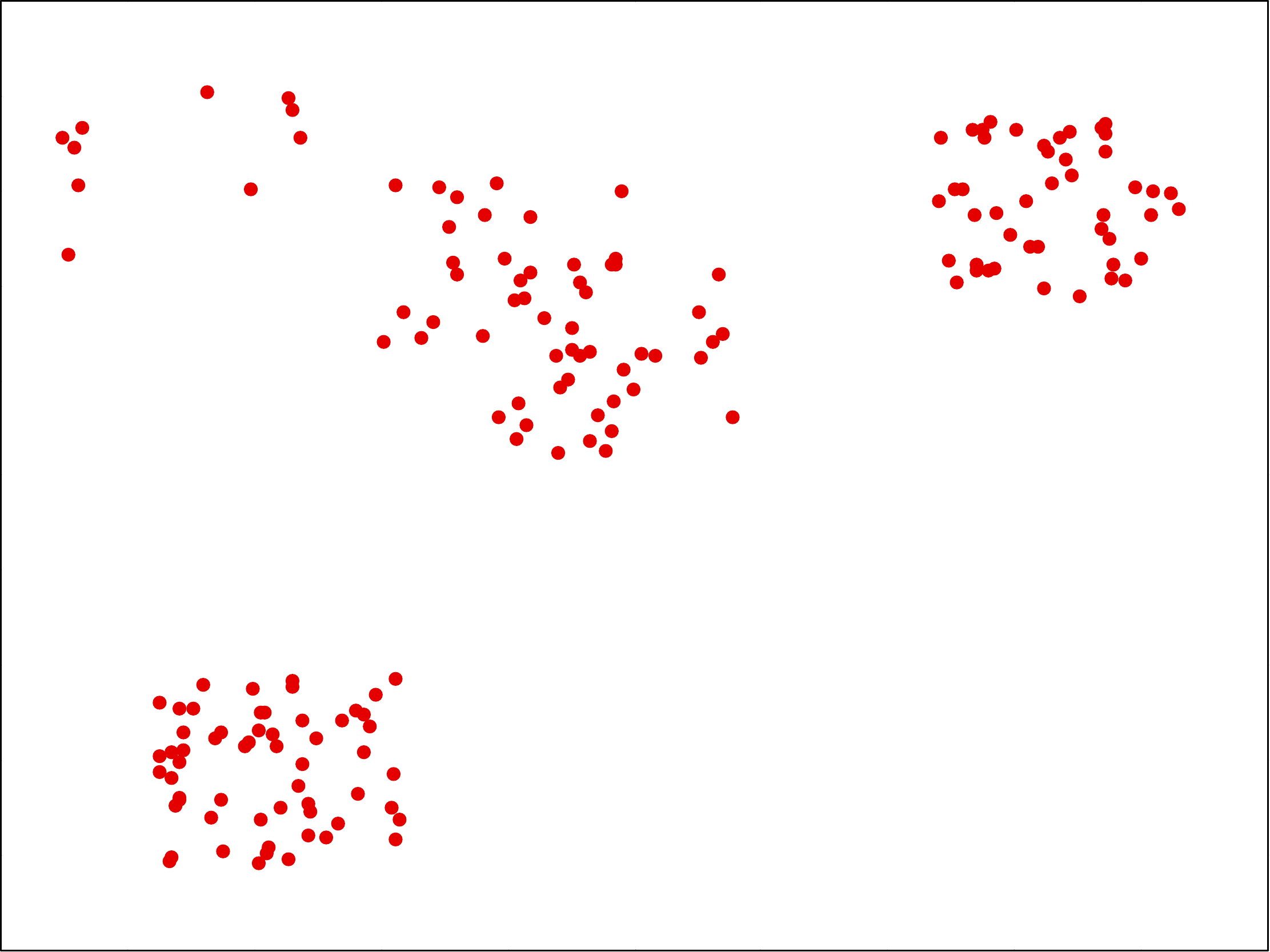}
	\caption{Two draws from a clustered point process \label{bind}}
      \end{wrapfigure}%
      \paragraph{Clustered point processes.}
Clustered point processes are useful in modelling phenomena which involve multiple
points spawning from individual seeds, such as clusters of trees, galaxies, or diseases.
Informally, a clustered point process is anything built using the monadic bind $\gamma_1\bind \lambda x\!. \gamma_2(x)$, where $\gamma_1$ is a point process for the initial seeds, and $\gamma_2$ is a point process that grows from each seed, where the location $x$ of a given seed may be a parameter. 
For a simple example, consider the point process in Fig.~\ref{bind} consisting of small square
clusters within the unit square. To simulate it we first sample the centres of these clusters from
our Poisson point process $\pi$  on the unit square~\eqref{eqn:poisson-sim}, and for each
cluster center we sample another Poisson point process. To sample the second point process, we again sample from another Poisson distribution $\underline N'$,
whose rate now depends on the provided coordinates of the cluster -- the closer to the
diagonal, the higher the rate, and we uniformly distribute these points in a small
square about this center using $\underline U'$, which is a location-dependent and 
scaled-down modification of $\underline U$ introduced earlier.
This results a Poisson number of clusters of Poisson processes, with those closer to the diagonal
being denser than those farther away.
\[
	\beta \ =\  \pi \bind_{GB} \lambda (x,y)\!.\ (\underline N' (x,y) \bind_{GB} \lambda \star\!.\ 
	\underline U' (x,y))\qquad\in GB\mathbb{I}^2. 
\]
Here, $\underline N$ is itself defined using the monad \eqref{eqn:poisson-sim}. This example is quite simple, but already illustrates that we can use the monad to quickly and clearly compose point processes to build complex statistical models. 


\section{The intensity measure as a monad morphism}
\label{sec:intensity}

A useful characteristic for describing
point processes is the expected number of points in a given region. For
example, in Fig.~\ref{example} we illustrated the homogeneous Poisson process with rate $10$.
The expected number of points in any region is proportional to $10a$ where $a$ is the area of the region. 
More generally, the \emph{intensity measure} of a point process is the measure that assigns to each measurable subset the expected number points in it.

There is a function $\mathbb E:GB(X)\to M(X)$ that takes a point process to its intensity measure.
In Theorem.~\ref{thm:expmonadmorph} we show that this function is actually a monad morphism from the point process monad (\S\ref{pp-monad}) to the monad of all measures (\S\ref{sec:allmeasuresmonad}). 
Thus, if we build a point process using the monadic constructions (for example by composing morphisms in the Kleisli category of $GB$) then we can immediately read off its intensity measure in a compositional way.
This new result generalizes Wald's lemma~(Ex.~\ref{cor:wald}), which is a core result in probability theory. 

\subsection{Constructing a monad morphism}

The expected number of points in a region $U \in \Sigma_X$ of a point process
$\alpha \in GBX$ can be given in terms of our generating sets $A^U_k$ for $BX$, as
$\sum_{k=1}^\infty k \cdot \alpha(A^U_k).$

We first show how to understand this in a more abstract way,
by injecting both probability measures $GX$ and bags $BX$ each
into measures $MX$, in a measurable and natural way.
The injection $i^G : G \rightarrow M$ is straightforward because $GX \subseteq MX$.
The injection $i^B : B \rightarrow M$ sends bags $[x_1,\dots,x_n]$ to  $\sum_{i=1}^n \delta_{x_i}$,
the multiplicity-respecting sum of Dirac deltas
centered around the elements $x_i$. It is measurable since its
inverse image map sends the generating sets $m^U_r \in \Sigma_{MX}$ to 
$\bigcup_{i=0}^{\left \lfloor{r}\right \rfloor} A^U_i \in \Sigma_{BX}$.
The proof that this is injective relies on $X$ being standard Borel. 
This injection is familiar in point process theory, indeed many authors actually define $BX$ as a space of integer-valued measures in the first place.
We can combine the horizontal composition of these two injections ($i^G \ast i^B$) with the
multiplication of $M$ in order to define $\mathbb{E}$.
$$\mathbb{E}\quad \defeq\quad  GB \xrightarrow{i^G \ast i^B} MM \xrightarrow{\mu^M} M $$
In the remainder of this paper we omit $\ast$ when writing horizontal compositions.
This definition of $\mathbb{E}$ does indeed
return the intensity measure of a point process:
\begin{lemma}\label{exp-lem}
For any point process $\alpha \in GBX$, $\mathbb{E}(\alpha)(U) = 
 \sum_k k \cdot \alpha(A^U_k)$.
\end{lemma}
\begin{proof}
Consider $\alpha \in GBX$ and $U \in \Sigma_X$. On expanding the horizontal composition
$i^Gi^B$ and using the change-of-variables formula for pushforward measures we have that
$$\mathbb{E}(\alpha)(U) = \int_{MX} \textrm{ev}_U\ \!\diff Mi^B_X(i^G_{BX}(\alpha))
= \int_{b \in BX} i^B_X(b)(U)\ \alpha(\diff b).$$
We separately compute this integral on the disjoint partitions $A^U_k$ ($k \in \mathbb N$) of $BX$.
In each partition, the value of $i^B_X(b)(U)$ is equal to $k$ (by definition).
This gives us the desired infinite
sum of $\sum_k k \cdot \alpha(A^U_k)$.
\end{proof}

\noindent
To show that $\mathbb{E} : GB \rightarrow M$ is a monad morphism we need to prove that
$$\text{\bf (Unit)}\ \ \eta^M = \mathbb{E} \circ \eta^{GB} \qquad\text{and}\qquad \mu^M \circ \mathbb{E}\mathbb{E} = \mathbb{E}\circ \mu^{GB}\ \ \text{\bf (Mult)}.$$

\noindent
Our main result stems from the fact that $l$ interacts well
with $i^G$ and $i^B$, which we prove next.
\begin{lemma}\label{l-well-behave}
{$(\mathbb{E} \circ l =)\ \mu^M \circ i^G\ast i^B \circ l = \mu^M \circ i^Bi^G : BG \rightarrow M$.}
\end{lemma}
\begin{wrapfigure}[5]{r}{0.25\textwidth}
\vspace{-3em}
\begin{tikzpicture}[scale=0.73]
    \node[shape=rectangle] (a) at (0,2) {$BGX$};
    \node[shape=rectangle] (b) at (4,2) {$GBX$};
    \node[shape=rectangle] (c) at (0,0) {$M^2X$};
    \node[shape=rectangle] (d) at (4,0) {$M^2X$};
    \node[shape=rectangle] (e) at (2,0) {$MX$};
    \draw [->](a) -- node[above] {$l_X$} (b);
    \path [->](a) edge node[right] {$(i^Bi^G)_X$} (c);
    \draw [->](c) -- node[below] {$\mu^M_X$} (e);
    \draw [->](d) -- node[below] {$\mu^M_X$} (e);
    \path [->](b) edge node[left] {$(i^Gi^B)_X$} (d);
\end{tikzpicture}
\end{wrapfigure}
\noindent\textit{Proof.}
(Diagram chasing) Consider $[\nu_1,\dots,\nu_n] \in BGX$ and $U \in \Sigma_X$.
Going from $BGX$ to $MX$ along the left edge and applying the resulting map to $U$ gives us
$\sum_i \nu_i(U)$. Along the other edge, making use of Lemma \ref{exp-lem}, we get
\mbox{$\sum_i i \cdot l[\nu_1,\dots,\nu_n](A^U_i)$}. Their equality can be
proved by noticing that $l[\nu_1,\dots,\nu_n](A^U_k)$ is simply
the coefficient of $x^k$ in the polynomial
$P(x) = \prod_i(\nu_i(\bar U) + \nu_i(U) \cdot x)$. And so equivalently
$P(x) = \sum_i l[\nu_1,\dots,\nu_n](A^U_i) \cdot x^i$. The desired equality is then
arrived at by taking the derivative of $P(x)$ at $x=1$.
\qed
\vspace{5pt}

\begin{lemma}\label{submonads}
$i^G : G \rightarrow M$ and $i^B : B \rightarrow M$ are monad morphisms.
\end{lemma}
\vspace{5pt}
\begin{theorem}The intensity measure $\mathbb{E} : GB \rightarrow M$ is a monad morphism.
  \label{thm:expmonadmorph}
\end{theorem}
\begin{proof}
A simple calculation shows {\bf Unit} to hold. For {\bf Mult} 
consider the two diagrams below.\\
\begin{tikzpicture}[scale=2,yscale=-0.75,xscale=1.75]
\begin{scope}
    \node[shape=rectangle] (a) at (0,0) {$GBGB$};
    \node[shape=rectangle] (b) at (2,0) {$GGBB$};
    \node[shape=rectangle] (c) at (1,1) {$GMB$};
    \node[shape=rectangle] (D) at (0,1) {$GMMB$};
    \node[shape=rectangle] (F) at (2,1) {$GMMB$};
    \node[shape=rectangle] (d) at (0,2) {$M^4$};
    \node[shape=rectangle] (f) at (2,2) {$M^4$};
    \node[shape=rectangle] (e) at (1,2) {$M^3$};
    \node[shape=rectangle] (g) at (0,3) {$M^3$};
    \node[shape=rectangle] (h) at (2,3) {$M^3$};
    \node[shape=rectangle] (i) at (1,3) {$M^2$};
    \node[shape=rectangle] (j) at (0,4) {$M^2$};
    \node[shape=rectangle] (l) at (2,4) {$M^2$};
    \node[shape=rectangle] (k) at (1,4) {$M$};
    \node[shape=circle, draw=gray] (x) at (1,0.5) {I};
    \node[shape=circle, draw=gray] (x) at (0.5,1.4) {II};
    \node[shape=circle, draw=gray] (x) at (1.5,1.4) {II};
    \node[shape=circle, draw=gray] (x) at (0.5,3) {III};
    \node[shape=circle, draw=gray] (x) at (1.5,3) {III};

    \draw [->](a) --   node[font=\scriptsize,above] {$GlB$} (b);
    \path [->](D) edge node[font=\scriptsize,right] {$i^GMMi^B $} (d);
    \path [->](F) edge node[font=\scriptsize,left] {$i^GMMi^B $} (f);
    \draw [->](D) --   node[font=\scriptsize,above] {$G\mu^MB$} (c);
    \draw [->](F) --   node[font=\scriptsize,above] {$G\mu^MB$} (c);
    \path [->](a) edge node[font=\scriptsize,right] {$Gi^Bi^GB $} (D);
    \path [->](b) edge node[font=\scriptsize,left] {$Gi^Gi^BB $} (F);
    \draw [->](d) --   node[font=\scriptsize,above] {$M\mu^MM$} (e);
    \draw [->](f) --   node[font=\scriptsize,above] {$M\mu^MM$} (e);
    \path [->](c) edge node[font=\scriptsize,left] {$i^GMi^B$} (e);
    \path [->](d) edge node[font=\scriptsize,right] {$\mu^MMM$} (g);
    \path [->](f) edge node[font=\scriptsize,left] {$\mu^MMM$} (h);
    \path [->](e) edge node[font=\scriptsize,right] {$M\mu^M $} (i);
    \path [->](g) edge node[font=\scriptsize,right] {$M\mu^M $} (j);
    \path [->](h) edge node[font=\scriptsize,left] {$M\mu^M$} (l);
    \path [->](i) edge node[font=\scriptsize,left ] {$\mu^M$} (k);
    \draw [->](j) --   node[font=\scriptsize,above] {$\mu^M$} (k);
    \draw [->](l) --   node[font=\scriptsize,above] {$\mu^M$} (k);
\end{scope}
\begin{scope}[shift={(2.75,0.5)},xscale=0.75,yscale=1.5]
    \node[shape=circle,draw=gray] (A) at (0.6,0.5) {IV};
    \node[shape=circle,draw=gray] (B) at (1.5,0.5) {V};
    \node[shape=circle,draw=gray] (C) at (0.6,1.5) {VI};
    \node[shape=circle,draw=gray] (D) at (1.55,1.5) {VII};
    \node[shape=rectangle] (1) at (0,0) {$GGBB $};
    \node[shape=rectangle] (2) at (1,0) {$GBB$};
    \node[shape=rectangle] (3) at (2,0) {$GB$};
    \node[shape=rectangle] (4) at (0,1) {$MMBB $};
    \node[shape=rectangle] (5) at (1,1) {$MBB $};
    \node[shape=rectangle] (6) at (2,1) {$MB$};
    \node[shape=rectangle] (7) at (0,2) {$M^4 $};
    \node[shape=rectangle] (8) at (1,2) {$M^3$};
    \node[shape=rectangle] (9) at (2,2) {$M^2$};
    \path [->](1) edge node[font=\scriptsize,right] {$i^Gi^GBB $} (4);
    \path [->](2) edge node[font=\scriptsize,right]  {$i^GBB $} (5);
    \path [->](3) edge node[font=\scriptsize,left]  {$i^GB $} (6);
    \path [->](4) edge node[font=\scriptsize,right] {$MMi^Bi^B $} (7);
    \path [->](5) edge node[font=\scriptsize,right]  {$Mi^Bi^B$} (8);
    \path [->](6) edge node[font=\scriptsize,left]  {$Mi^B $} (9);
    \draw [->](1) --   node[font=\scriptsize,above] {$\mu^GBB $} (2);
    \draw [->](2) --   node[font=\scriptsize,above] {$G\mu^B $} (3);
    \draw [->](4) --   node[font=\scriptsize,above] {$\mu^MBB$} (5);
    \draw [->](5) --   node[font=\scriptsize,above] {$M\mu^B$} (6);
    \draw [->](7) --   node[font=\scriptsize,above] {$\mu^MMM$} (8);
    \draw [->](8) --   node[font=\scriptsize,above] {$M\mu^M$} (9);
\end{scope}
\end{tikzpicture}

\noindent
All the sub-diagrams above commute: (I) due to Lemma \ref{l-well-behave}, (II) by
naturality, (III) due to associativity of $\mu^M$, (IV) and (VII) due to $i^G$ and
$i^B$ being monad morphisms (Lemma \ref{submonads}),
and finally (V) and (VI) by naturality. Using the commutative diagrams above we
prove the required equality.
\begin{align*}
\mu^M \circ \mathbb{E}\mathbb{E}
 &\ :\  GBGB \rightarrow M
 \\
\mu^M \circ \mathbb{E}\mathbb{E}
 &= \mu^M \circ M\mu^M \circ \mu^MMM \circ i^GMMi^B \circ {Gi^Bi^GB}
&\text{(defn. of $\mathbb{E}$ + naturality)}
 \\
 &= \mu^M \circ M\mu^M \circ \mu^MMM \circ {i^GMMi^B \circ Gi^Gi^BB} \circ GlB
&\text{(left diagram)}
 \\
 &= \mu^M \circ {M\mu^M \circ \mu^MMM \circ MMi^Bi^B \circ i^Gi^GBB} \circ GlB
&\text{(naturality)}
 \\
 &= \mu^M \circ Mi^B \circ i^GB \circ G\mu^B \circ \mu^GBB \circ GlB
&\text{(right diagram)}
 \\
 &= \mu^M \circ Mi^B \circ i^GB \circ \mu^{GB}
&\text{(defn. of $\mu^{GB}$)}
 \\
 &= \mathbb{E} \circ \mu^{GB}
&\text{(defn. of $\mathbb{E}$)}
\end{align*}
\vspace{-1.2cm}

\end{proof}

\subsection{Examples}
\begin{example}In 
  \S\ref{pois-example} we simulated a Poisson point process by composing the Poisson distribution with a uniform singleton. We show this has the required intensity measure in a compositional way, using the monad morphism. Let $\underline N \in GB\mathbbm 1\cong G\mathbb N$ be the Poisson distribution with mean
  $\Lambda$, and let $\underline U\in GB\mathbb I^2$ be the uniformly distributed single-point
  process, from $\S\ref{sec:examples}$.
The simulated Poisson process is $\pi=(\underline{N}\  \bind_{GB} \lambda \star\!.\  \underline{U})$, and we have 
\begin{align*}
\mathbb E(\pi) &= \mathbb E(\underline{N}\  \bind_{GB} \lambda \star\!.\  \underline{U})\\
&= \mathbb E(\underline{N}) \bind_M \lambda \star\!.\  \mathbb E(\underline{U})&\text{(Theorem~\ref{thm:expmonadmorph})}\\
&= \lambda W \in \Sigma_{\mathbb{I}^2}. \mathbb E(\underline{N})(\star) \times \mathbb E(\underline{U})(W)\\
&= \lambda W \in \Sigma_{\mathbb{I}^2}. \Lambda \times \abs{W} \qquad\qquad\in M\mathbb{I}^2  
\end{align*}
In the penultimate step we use the fact that in the discrete case the bind for $M$, just like
for $G$, computes a weighted sum, which in this case is just a single term.
In the final step, we substitute in the intensity of $\underline N$
and the intensity measure of the uniform point process $\underline U$, 
producing the correct intensity. $W$ is a measurable subset of the unit square 
and $\abs{W}$ is its area.
\end{example}

\begin{example}[Discrete Wald's Lemma]
  We regard arbitrary probability distributions $\underline N,\underline X$ on the natural numbers as point
  processes in $GB\mathbbm 1$ via \eqref{eqn:GB1GBN}.
  Wald's lemma says that the expected value of the compound distribution~($\gamma$, \eqref{eqn:compounddist}) is the product of the expectations, which is immediate from the fact that $\mathbb E$ is a monad morphism:
\label{cor:wald}
\begin{align*}
\mathbb E(\gamma) &= \mathbb E(\underline{N} \bind_{GB} \lambda \star\!.\  \underline{X})\\
&= \mathbb E(\underline{N}) \bind_{M} \lambda \star\!.\  \mathbb E(\underline{X})& \text{(Theorem~\ref{thm:expmonadmorph})}\\
&= \lambda \star\!.\  \mathbb E(\underline{N})(\star) \times \mathbb E(\underline{X})(\star) \qquad\in M\mathbbm 1 
\end{align*}
\end{example}

\paragraph{Remark.}
We remark that a natural transformation in the \emph{opposite} direction ($M^+ \to GB$) has been exhibited in~\cite{giry-machine}, where $M^+(X)$ is the space of finite non-empty measures. This natural transformation takes an intensity measure to the corresponding inhomogeneous Poisson process. Since $M^+$ is not a monad, it remains to be seen whether this natural transformation can be made into a monad morphism somehow.


\paragraph{Concluding remarks.}
We have exhibited a monad $GB$ for point processes (\S\ref{pp-monad}), and shown that the intensity measure is a monad morphism~(\S\ref{sec:intensity}). This gives a compositional way of building and reasoning about increasingly complicated point processes (\S\ref{sec:examples}). This is further evidence towards the claim that applied category theory has the potential to be a useful tool for statistical modelling. 

\paragraph{Acknowledgements.}
We are grateful for discussions with Peter Lindner regarding the role of point processes in his work~\cite{grohe-lindner}. Thanks too to Bart Jacobs and Gordon Plotkin about the role of multisets.
Thanks to the anonymous reviewers and to Mathieu Huot and Dario Stein for their feedback.
Finally we appreciate the opportunity to present this work at the LAFI 2020 workshop~\cite{lafi2020}.
Staton's research is supported by a Royal Society University Research Fellowship.

\bibliographystyle{eptcsini}
\bibliography{generic}
\end{document}